\documentclass[a4paper,11pt,dvipdfmx]{article}
\usepackage{amsmath,amssymb,amsthm}
\usepackage{amsfonts}
\usepackage{eucal}
\usepackage{color}
\numberwithin{equation}{section}

\renewcommand{\a}{\alpha}

\renewcommand{\c}{\gamma}
\renewcommand{\d}{\delta}
\newcommand{\e}{\varepsilon}
\newcommand{\f}{\varphi}
\newcommand{\g}{\psi}

\newcommand{\m}{\mu}

\newcommand{\s}{\sigma}
\newcommand{\x}{\xi}
\newcommand{\y}{\eta}
\newcommand{\z}{\zeta}

\renewcommand{\L}{\Lambda}
\newcommand{\co}{\mathbb{C}}
\newcommand{\re}{\mathbb{R}}
\newcommand{\ze}{\mathbb{Z}}

\renewcommand{\O}{\Omega}

\def\pa{\partial}

\newcommand{\supp}{\mathrm{{supp}}}

\newcommand{\edges}{\mathcal{E}}
\newcommand{\fourier}{\mathcal{F}}

\newcommand{\bra}[1]{\{#1\}}

\newcommand{\bigpare}[1]{\bigl(#1\bigr)}
\newcommand{\biggpare}[1]{\biggl(#1\biggr)}

\newcommand{\bigset}[2]{\bigl\{#1\bigm|#2\bigr\}}

\newcommand{\biggset}[2]{\biggl\{#1\biggm|#2\biggr\}}

\newcommand{\norm}[1]{\| #1 \|}
\newcommand{\bignorm}[1]{\bigl\| #1 \bigr\|}

\newcommand{\bigabs}[1]{\bigl| #1 \bigr|}

\newcommand{\jap}[1]{\langle #1 \rangle}

\newtheorem{thm}{Theorem}[section]
\newtheorem{lem}[thm]{Lemma}

\newtheorem{cor}[thm]{Corollary}

\theoremstyle{definition}

\newtheorem{ass}{Assumption}

\theoremstyle{remark}
\newtheorem{rem}{Remark}[section]

\newtheorem{example}{Examples}[section]

\title{Continuum limit for Laplace and Elliptic operators on lattices}
\author{
Keita Mikami\footnote{iTHEMS, RIKEN, 2-1 Hirosawa, Wako, Saitama, 351-0198, Japan, 
Email: keita.mikami@riken.jp}
\and
Shu Nakamura\footnote{Department of Mathematics, Faculty of Sciences, Gakushuin University, 1-5-1, Mejiro, Toshima, Tokyo, Japan 171-8588, Email: shu.nakamura@gakushuin.ac.jp}
\and
Yukihide Tadano\footnote{Department of Mathematics, Faculty of Science Division I, Tokyo University of Science, 1-3,Kagurazaka, Shinjuku, Tokyo, Japan 162-8601, E-mail: y.tadano@rs.tus.ac.jp}
}
\begin{document}
\maketitle

\begin{abstract}
Continuum limits of Laplace operators on general lattices are considered, and it is shown that these 
operators converge to elliptic operators on the Euclidean space in the sense of the generalized norm 
resolvent convergence. We then study operators on the hexagonal lattice, which does not apply the above general 
theory, but we can show its Laplace operator converges to the continuous Laplace operator in the continuum limit. We also study discrete operators on the square lattice corresponding to second order strictly elliptic operators with variable coefficients, and prove the generalized norm resolvent convergence in the continuum limit. 
\end{abstract}

%%%%%%%%%%%%%%%%%%%%%%%% SECTION 1%%%%%%%%%%%%%%%%%%%%%%%%%%
\section{Introduction}\label{sec-intro}

In \cite{NaTad}, the continuum limit of Schr\"odinger operators on the square lattice was studied, 
and the generalized norm convergence was proved under rather general conditions on the 
perturbations. After the publication, several generalizations of the result have been appeared, 
for example the continuum limit of discrete Fourier multipliers \cite{C-G-J}, 
discrete Dirac operators \cite{C-G-J-2}, discrete operators on the half space \cite{C-G-J-3}, 
discrete Hodge-Dirac operators \cite{Miranda-Parra}, and so on. 
However, these researches concern discrete operators on the square lattice only (however, see 
\cite{Isozaki-Jensen}). In this note, 
we first investigate a generalization to operators on more general lattices, including the triangular lattice, 
the tetrahedral lattice etc. The main novelty in this part is a generalized formulation of the model, 
but most arguments are modified to accommodate the more general setting, and applications 
to Laplace operators on several lattices are also discussed. 

Then we discuss the continuum limit of the hexagonal lattice, which does not satisfy the usual definition 
of the \emph{lattices}\/ above. However, with some additional constructions and modifications, 
we can prove that the continuum limit of the Laplace operator on hexagonal lattice is in fact 
the usual Laplace operator on $\re^2$, and the convergence holds in the sense of the generalized
norm convergence. 

Finally we discuss discrete operators on the square lattice corresponding to a second order strictly elliptic operator on $\re^d$. We show that in the continuum limit, these operators also converge to the elliptic operator in the generalized norm resolvent sense. 

In mathematical physics, discrete Schr\"odinger operators on lattices (with random potentials)  have been studied extensively in the context of random Schr\"odinger operators (see, e.g., \cite{Aizenman-Warzel} and references therein). Discrete Schr\"odinger operators have also been studied in the context of the scattering theory, mostly in the case of square lattices (see, e.g., \cite{BdM-Sa}). 
However, Ando and others studied spectral properties and inverse scattering theory for various lattices (\cite{Ando13}, \cite{AIM-1}, \cite{AIM-2}), and they are relevant to the current work (see also Tadano \cite{Tadano19} for the long-range scattering on the hexagonal lattice). 
In solid state physics, such discrete Schr\"odinger operators are called \emph{the tight binding model}, and used to approximate quantum physics of electrons in solids (see, e.g., \cite{AshMer}, Chapter~10 for the physics background). In particular, we refer Semenoff \cite{Sem84} for the analysis of the electron in the hexagonal lattice. 
On the other hand, discrete \emph{Dirac}\/ operators were extensively used in the lattice field theory, and the continuum limit has been, and still is, a central subject of the research (see, e.g., \cite{Rothe} for an introduction to the lattice field theory. See also \cite{C-G-J-2}, \cite{Schmidt-Umeda},
\cite{Miranda-Parra} and \cite{Na23} for recent development on rigorous results). 
As mentioned above, mathematically rigorous proof of the continuum limit of discrete Schr\"odinger operators were scarce, and several results appeared only recently (\cite{NaTad}, \cite{C-G-J}, \cite{C-G-J-2}, \cite{Isozaki-Jensen}. See also \cite{ExNaTad}). In this note, we extend the results of Nakamura-Tadano \cite{NaTad} in various directions. The proofs are based on the same ideas as \cite{NaTad}, but we need generalized formulations and several new ideas for each cases. 

This note is constructed as follows: In Section~2, we explain our main results. In Subsection~2.1, we give an explicit formulation of general lattices and state the convergence theorem for the continuum limit in Subsection~2.2. In Subsection~2.3, we give the definition of the hexagonal lattice and the Schr\"odinger operator on it, and state the convergence theorem. We then describe our elliptic operator on the square lattice and state the convergence theorem in Subsection~2.4. 
Section~3 is devoted to the discussion of the general lattice. We give detailed construction of the embedding operator and the proof of main theorems. The argument is similar to \cite{NaTad}, but the formulation is substantially generalized. We also give explicit examples: The triangular lattice in $\re^2$; the tetrahedral lattice and the octahedral lattice in $\re^3$, and compute the limiting operators. 
In Section~4, we prove the convergence theorem for the hexagonal lattice. The embedding operator is slightly more complicated in this case, and we use, in addition to the standard 2-component representation of the hexagonal lattice, a 3-component representation of the triangular lattice which includes the hexagonal lattice. We remark that the convergence rate (for the free operators) is $O(h^2)$, which is the same as the square lattice case, and is optimum. Section~5 is the proof of convergence for the continuum limit of discrete elliptic operators. Here the key lemma is an elliptic estimate for the discrete elliptic operators, Lemma~\ref{lem-elliptic-estimate}. 

\noindent
\textbf{Acknowledgement:} SN was partly supported by JSPS Kakenhi Grant Number 21K03276. YT was partly supported by JSPS Kakenhi Grant Number 23K12991.

%%%
\section{Main results} \label{sec-main}
\subsection{Lattices and graphs, Laplace operators}\label{subsec-intro-defs}
A set $\L\subset\re^d$, $d\geq 1$, is called a \textit{lattice}\/ 
if there is a non-singular $d\times d$-matrix $L$ such that $\L=L\ze^d$. 
A lattice $\L$ has a natural Abelian group structure, and acts on $\re^d$ by addition. 
We denote the volume of the fundamental domain by $\omega_0>0$, i.e., $\omega_0=|\det L|$. 
This definition applies to the triangular lattice, the tetrahedral lattice, the octahedral lattice, etc., 
but does not apply to the hexagonal lattice, the kagome lattice, etc. 

In order to define a Laplace operator on a lattice $\L$, we need to introduce a graph structure. 
Let $\mathbf{G}=(\L,\edges)$ be a connected graph where the set of vertices is $\L$ and the set of the directed 
edges is denoted by $\edges$. We suppose the graph is invariant under the action of $\L$, 
i.e., $[a,b]\in\edges$ if and only if $[a+x,b+x]\in \edges$ for all $x\in\L$.  We also suppose 
$[a,b]\in\edges$ if and only if $[b,a]\in\edges$.

Let $\m$ : $\edges \to (0,\infty)$ be a weight function, and we suppose $\m$ is invariant 
under the action of $\L$, i.e., $\m([a,b])=\m([a+x,b+x])$ for all $[a,b]\in\edges$, $x\in\L$, 
and symmetric in the sense $\m([a,b])=\m([b,a])$ for $e=[a,b]\in\edges$. 
We denote the initial vertex of an edge $e$ by $i(e)$, and the terminal vertex by $t(e)$, i.e., 
$i(e)=a$ and $t(e)=b$ when $e=[a,b]$. 

Then, for $u\in\ell^2(\L)$, the Laplace operator is defined by 
\[
-\triangle_\L u(x)=\sum_{i(e)=x} \m(e)(u(x)-u(t(e))). 
\]
It is easy to see $-\triangle_\L$ 
is a bounded positive operator on $\ell^2(\L)$. 

We note we can write 
$\bigset{e\in\edges}{i(e)=0} =\{[0,\pm f^j]$, $j=1,\dots, K\}$
with $f^1,\dots,f^K\in\L$. Since $(\L,\edges)$ is connected, 
$\mathrm{Span}(\{f^j\}_{j=1}^K)=\re^d$. We denote
\[
\m_j=\m([0,f^j])=\m([0,-f^j]), \quad j=1,\dots, K.
\]
Then the Laplace operator can be written in a more familiar form: 
\[
-\triangle_\L u(x) = \sum_{j=1}^K \m_j (2u(x)-u(x+f^j)-u(x-f^j)).
\]

\subsection{Continuum limit and the norm resolvent convergence}\label{subsec-intro-cont-lim}
Now we introduce a scaling in order to study the continuum limit. 
Let $h>0$ and we set 
\[
\mathbf{G}_h =(h \L, h\edges)\ \text{ where } \ h\L=\bigset{hx}{x\in\L}, \ h\edges=\bigset{[ha,hb]}{[a,b]\in\edges}.
\]
The norm of $\ell^2(h \L)$ is given by 
\[
\|u\|_{\ell^2(h\L)}^2 = \omega_0 h^d \sum_{x\in h\L} |u(x)|^2, 
\quad u\in\ell^2(h\L). 
\]
The scaled Laplace operator is defined by 
\[
H_{0,h} u(x) = h^{-2} \sum_{j=1}^K \m_j (2u(x)-u(x+h f^j)-u(x-hf^j))
\]
for $u\in\ell^2(h\L)$. We consider the behavior of $H_{0,h}$ as $h\to 0$.

In order to state our main result, we set a differential operator on $L^2(\re^d)$ 
to which $H_{0,h}$ converges as $h \to 0$. We set a quadratic form 
\[
p_0(2\pi\x)= \sum_{j=1}^K \m_j (f^j\cdot 2\pi\x)^2, \quad \x\in\re^d.
\]
Since $\{f^j\}_{j=1}^K$ spans $\re^d$, this is a non degenerate positive quadratic form, and we set
\[
H_0 = p_0(D_x) =\sum_{j=1}^K \m_j (f^j\cdot D_x)^2\quad \text{on }\re^d, 
\]
where $D_x=(-i\pa_{x_1}, \dots, -i\pa_{x_d})$. 
We note that $H_0$ is a second order elliptic operator with constant coefficients, and 
hence $H_0$ is self-adjoint with $\mathcal{D}(H_0)=H^2(\re^d)$. It is also easy to show 
its spectrum $\s(H_0)=[0,\infty)$ and is absolutely continuous. 

We will construct an isometry $J_h$ : $\ell^2(h\L) \to L^2(\re^d)$ in 
Subection~\ref{sec-const-J-lattice}.
Our first result claims $H_{0,h}$ converges to $H_0$ as $h\to 0$
in the generalized norm convergence. 

\begin{thm}\label{thm-free-hamiltonian}
For any $\m\in\co\setminus\re$, there is $C>0$ such that 
\[
\bignorm{J_h (H_{0,h}-\m)^{-1}J_h^* -(H_0-\m)^{-1}}_{\mathcal{B}(L^2(\re^d))} \leq Ch^2, 
\quad h>0. 
\]
\end{thm}

This result, combined with Lemma~\ref{lem-1-JJ}, also implies the following: 
\[
\bignorm{J_h^*(H_0-\m)^{-1}J_h - (H_{0,h}-\m)^{-1}}_{\mathcal{B}(\ell^2(h\L))} \leq Ch^2, 
\quad h>0. 
\]

We also consider perturbed operators
\[
H=H_0+V(x) \quad \text{on } L^2(\re^d),
\]
where $V(x)$ is a real-valued function, and we define
\[
H_h=H_{0,h}+V_h, \quad\text{on }\ell^2(h\L),  
\]
where $V_h(z)=V(z)$ for $z\in h\L$. 

We suppose $V$ satisfies the same assumption as in \cite{NaTad}. 
\begin{ass}\label{ass-V}
$V$ is a real-valued continuous function on $\re^d$, and bounded from below. 
$(V(x)+M)^{-1}$ is uniformly continuous with some $M>0$, and there is $c_1>0$ such that 
\[
c_1^{-1} (V(x)+M) \leq V(y)+M\leq c_1(V(x)+M), \quad \text{if }|x-y|\leq 1. 
\]
\end{ass}
The above assumption implies $H$ is essentially self-adjoint on $C_c^\infty(\re^d)$.
%For simplicity, we suppose $V$ is a bounded uniformly continuous function on $\re^d$. 
%Then it is easy to see $H$ is self-adjoint with $\mathcal{D}(H)=H^2(\re^d)$. 

\begin{thm}\label{thm-hamiltonian-with-potential}
Suppose Assumption~\ref{ass-V}, 
and let $H$ and $H_h$ be defined as above. 
Then for any $\m\in\co\setminus\re$, 
\[
\bignorm{J_h (H_{h}-\m)^{-1}J_h^* -(H-\m)^{-1}}_{\mathcal{B}(L^2(\re^d))} \to 0
\quad \text{as } h\to 0. 
\]
\end{thm}

\begin{cor}\label{cor-spectrum-lattice}
Under the assumptions of Theorem~1.2, 
\[
\mathbf{d}_H(\s(H_h),\s(H)) \to 0, \quad\text{as } h\to 0, 
\]
where $\mathbf{d}_H(\cdot,\cdot)$ denotes the Hausdorff distance on the subsets of $\re$, 
and $\s(A)$ denotes the spectrum of an operator $A$. 
\end{cor}

\subsection{The hexagonal lattice}\label{subsec-intro-hex}
Now we consider the hexagonal lattice defined as follows. 
We set 
\[
e_1=(1/2,\sqrt{3}/2),\quad e_2=(1/2,-\sqrt{3}/2)\quad \text{and}\quad e_3=(-1,0)
\]
in $\re^2$. 
We also set 
\[
f_1=e_1-e_3 =(3/2,\sqrt{3}/2)\quad \text{and}\quad f_2=e_2-e_3=(3/2,-\sqrt{3}/2).
\] 

%%% Figure 1 removed %%%

We write 
\[
\Gamma_h=\ze(h f_1)+\ze(h f_2)=\bigset{n_1hf_1+n_2hf_2}{n_1,n_2\in\ze},
\] 
where $h>0$ is the lattice spacing. 
We consider a graph $\mathbf{G}^{\mathrm{hex}}_h=(\Lambda^{\mathrm{hex}}_h, \mathcal{E}^{\mathrm{hex}}_h)$ embedded in $\re^2$, where 
\[
\Lambda^{\mathrm{hex}}_h=\Lambda^+_h+\Lambda^-_h
=(0+\Gamma_h)\cup (he_3+\Gamma_h),
\]
is the set of vertices, and 
\[
\mathcal{E}^{\mathrm{hex}}_h=\bigcup_{j=1,2,3} \bigset{[0,he_j]+\c}{\c\in\Gamma_h}
\cup \bigcup_{j=1,2,3}\bigset{[0,-he_j]+(he_3+\c)}{\c\in\Gamma_h}
\]
is the set of directed edges. $\mathbf{G}^{\mathrm{hex}}_h$ is our hexagonal lattice, and $\Gamma$ 
is the symmetry group. The Laplace operator is defined as follows: 
\[
H^{\mathrm{hex}}_h u(z)= \begin{cases} h^{-2}\sum_{j=1}^3 (u(z)-u(z+he_j)) \quad & \text{if }z\in \Lambda_h^+, \\
h^{-2}\sum_{j=1}^3 (u(z)-u(z-he_j)) \quad & \text{if }z\in \Lambda_h^-, \end{cases}
\]
where $u\in\ell^2(\Lambda^{\mathrm{hex}}_h)$. 
We also set
\[
H^{\mathrm{hex}}_0=-\frac{3}{4}\triangle \quad \text{on }L^2(\re^2).
\]
We will construct an identification operator $J^{\mathrm{hex}}_h : \ell^2(\Lambda^{\mathrm{hex}}_h)\to L^2(\re^2)$ in Subsection~\ref{subsec-intro-hex} such that
$J^{\mathrm{hex}}_h$ is asymptotically isometry in some sense (see Lemma~\ref{lem-hex-key1}). 
Then we can show 

\begin{thm}\label{thm-hexa}
For $\m\in\co\setminus\re$, 
\[
\norm{(H^{\mathrm{hex}}_0-\m)^{-1}-J^{\mathrm{hex}}_h (H^{\mathrm{hex}}_h-\m)^{-1} 
(J^{\mathrm{hex}}_h)^*}_{\mathcal{B}(L^2(\re^2))}\leq Ch^2 \quad \text{with }C>0, 
\]
and 
\[
\norm{(J^{\mathrm{hex}}_h)^*(H^{\mathrm{hex}}_0-\m)^{-1}J^{\mathrm{hex}}_h
- (H^{\mathrm{hex}}_h-\m)^{-1}}_{\mathcal{B}(\ell^2(\Lambda^{\mathrm{hex}}_h))}\leq Ch^2 \quad \text{with }C>0. 
\]
\end{thm}

%%%%%%%%%
\subsection{Elliptic operators on the square lattice}\label{subsec-intro-elliptic}

We now consider an elliptic operator $P$ on $L^2(\re^d)$ defined by
\[
Pu(x)=\sum_{j,k=1}^d D_j a_{j,k}(x) D_k u(x)+V(x)u(x), \quad \text{for }u\in C_c^2(\re^d),
\]
where $D_j =-i\pa/\pa x_j$, and $\{a_{j,k}(x)\}_{j,k=1}^d$ are sufficiently smooth coefficients. 
We always suppose they are symmetric: $a_{j,k}(x)=a_{k,j}(x)$, real-valued, and $V(x)$ is real-valued
so that $P$ is symmetric. 

\begin{ass}\label{ass-elli}
We assume $a_{j,k}\in C^1(\re^d)$ for each $j,k=1,\dots, d$. 
$a_{j,k}(x)$, $\pa_x a_{j,k}(x)$ are bounded on $\re^d$. 
Moreover, $\{a_{j,k}(x)\}_{j,k=1}^d$ is uniformly elliptic, i.e., there is $c_0>0$ such that
\[
\sum_{j,k=1}^d a_{j,k}(x) \x_j \x_k \geq c_0|\x|^2, \quad x, \x\in\re^d.
\]
\end{ass}

A corresponding operator on the square lattice $h\ze^d$, with the lattice spacing $h>0$, may be define as follows. Let $D_{h;j}^\pm$ be the forward/backward difference operators defined by 
\[
D_{h;j}^+u(x) =\frac1{ih}(u(x+he_j)-u(x)), \quad 
D_{h;j}^-u(x) =\frac1{ih}(u(x)-u(x-he_j)), 
\]
for $u\in \ell^2(h\ze^d)$, $j=1,\dots, d$ and $x\in h\ze^d$, 
where $\{e_j\}_{j=1}^d$ is the standard basis of $\re^d$, i.e., 
$(e_j)_k=\d_{j,k}$, $j, k=1,\dots,d$. We note $(D_{h;j}^\pm)^*=D_{h;j}^\mp$. 
We also denote $V_hu(x)=V(x)u(x)$, $x\in h\ze^d$. 
Then we set 
\[
P_h^\pm u(x) =\sum_{j,k=1}^d D_{h;j}^\mp a_{j,k}(x) D_{h;k}^\pm u(x)
+V_hu(x) , 
\quad x\in h\ze^d.
\]

\begin{rem}
$P_h^\pm$ are slightly more complicated than the definition may suggest. 
If we compute them explicitly, we have
\begin{align*}
P_h^\pm u(x) &= \pm\frac1{ih} \sum_{j,k=1}^d D_{h;j}^\mp\bigpare{ a_{j,k}(x)(u(x\pm he_k)-u(x))}
+V_hu(x)\\
&=- \frac{1}{h^2}\sum_{j,k=1}^d \bigpare{ a_{j,k}(x)u(x\pm he_k) -a_{j,k}(x) u(x) \\
&\quad -a_{j,k}(x\mp h e_j)u(x\pm he_k\mp he_j)+a_{j,k}(x\mp h e_j)u(x\mp he_j)}+V_hu(x).
\end{align*}
If $a_{j,k}=0$ for $j\neq k$, it becomes
\begin{align*}
P_h^\pm u(x) 
&= \frac{1}{h^2}\sum_{j=1}^d \bigpare{  (a_{j,j}(x)+a_{j,j}(x\mp h e_j)) u(x) \\
&\qquad - a_{j,j}(x)u(x\pm h{e_j})-a_{j,j}(x\mp h e_j)u(x\mp he_j)}+V_hu(x)\\
&= \frac{1}{h^2}\sum_{j=1}^d \bigpare{  (a_{j,j}(x)(u(x)-u(x\pm h{e_j})) \\
&\qquad +a_{j,j}(x\mp h e_j)(u(x)-u(x\mp he_j))}+V_hu(x),  
\end{align*}
and the coefficients $a_{j,j}(x)$ may be considered as a function on the set of edges, 
specifically $a_{j,j}(x)= \tilde a((x,x+h e_j))$ (for $+$-case). In particular, we have
\[
\jap{u, P_h^+ u} = 2\sum_{e : \text{edges}} \tilde a(e) |u(i(e)-u(t(e))|^2 +\sum_{x\in h\re^d} V(x)|u(x)|^2, 
\]
where $i((a,b))=a$, $t((a,b))=b$. 
But this is a special case, and in general we cannot use this expression. 
\end{rem}

\begin{thm}\label{thm-elliptic-convergence}
Suppose Assumptions~\ref{ass-V}~and~\ref{ass-elli}, and $(V(x)+M)^{-1}$ is Lipschitz continuous with some $M>0$. Then for any $z\in\co\setminus \re$ and $0<\a<1$, 
there is $C>0$ such that
\[
\bignorm{(P-z)^{-1} - J_h (P_h^\pm-z)^{-1} J_h^* } \leq Ch^\a, 
\quad h\in (0,1],
\]
where $J_h : \ell^2(h\ze^d)\to L^2(\re^d)$ is the identification operator as in Section~\ref{sec-standard-lattices}.  
\end{thm}

%%%%%%%%%%%%
\section{Standard lattices}\label{sec-standard-lattices}

\subsection{Dual lattices and the Fourier transform}

Concerning the Fourier transform, we use the notation
\[
\fourier u(\x) = \int_{\re^d} e^{-2\pi i x\cdot\x} u(x)dx, \quad \x\in\re^d 
\]
for a function (or distribution) $u(x)$ on $\re^d$. 

The dual lattice $\L'\subset\re^d$ is defined by
\[
\L'=\bigset{\x\in\re^d}{\x\cdot x\in \ze \text{ for all } x\in\\L}, 
\]
and it plays a crucial role in the following argument. The dual lattice may be written as 
$\L'={}^t L^{-1}\ze^d$ if $\L=L\ze^d$, and the volume of an arbitrary fundamental domain of $\L'$, described below,
is given by $\omega_0^{-1}=|\det L|^{-1}$. 

Let $\O\subset\re^d$ be a fundamental domain with respect to $\L'$, i.e., 
\[
\O\cap (\O+e) \text{ is measure zero for any $e\in \L'\setminus\{0\}$}, 
\quad \bigcup_{e\in\L'} (\O+e) =\re^d.
\]
For example, let $\{b_j\}_{j=1}^d \subset \L'$ be a basis of $\L'$, i.e., 
$\L'=b_1\ze+\cdots +b_d\ze$,  and in the following we may set
\[
\O=\biggset{\sum_{j=1}^d  t_j b_j}{0\leq t_j\leq 1, j=1,\dots, d}.
\]

If we identify $v\in \ell^2(h\L)$ with 
\[
\tilde v(x) = \omega_0 h^d \sum_{y\in h\L} v(y)\d(x-y), 
\]
then we have
\[
\fourier \tilde v(\x) =\omega_0 h^d \sum_{y\in h \L} e^{-2\pi i y\cdot\x}v(y).
\]
We note that $\fourier \tilde v$ is periodic with respect to the action of the dual lattice: 
$(h\L)'=h^{-1}\L'$. 
Motivated by these observations, for $v\in\ell^2(h \L)$, we define 
\[
F_h v(\x)=\omega_0 h^d \sum_{y\in h \L} e^{-2\pi i y\cdot\x}v(y), 
\quad \x\in h^{-1}\O.
\]
$F_h$ is a unitary from $\ell^2(h\L)$ to $L^2(h^{-1}\O)$, 
and the inverse is given by 
\[
F_h^* w(x) =\int_{h^{-1}\O} e^{2\pi i x\cdot\x} w(\x)d\x, \quad x\in h\L,
\]
for $w\in L^2(h^{-1}\O)$. We recall $\mathrm{vol}(h^{-1}\O)= \omega_0^{-1}h^{-d}$.

%%%%%%%%%%%
\subsection{The construction of identification operators}\label{sec-const-J-lattice}
In order to construct $J_h$, we need to find a function $\f\in\mathcal{S}(\re^d)$ such that 
$\{ \f(\cdot-z)\}_{z\in\L}$ is an orthonormal system. 
This condition is equivalent to the claim $\{ e^{-2\pi i \x\cdot z}\hat\f(\x)\}_{z\in\L}$ 
is an orthonormal system, where $\hat \f =\fourier\f$.  
We compute the inner product of these functions: 
\begin{align*}
&\int_{\re^d} \overline{e^{-2\pi i \x\cdot z}\hat\f(\x)} e^{-2\pi i \x\cdot y}\hat\f(\x) d\x
=\int_{\re^d} e^{2\pi i(z-y)\cdot\x} |\hat\f(\x)|^2 d\x \\
&\quad =\sum_{\y\in\L'}\int_{\O} e^{2\pi i(z-y)\cdot(\x+\y)} |\hat\f(\x+\y)|^2 d\x 
=\int_{\O} e^{2\pi i(z-y)\cdot\x} \sum_{\y\in\L'}|\hat\f(\x+\y)|^2 d\x.
\end{align*}
Here we have used the fact $(z-y)\cdot\y\in\ze$ if $z,y\in \L$ and $\y\in\L'$. 
Since $\{e^{2\pi i z\cdot\x}\}_{z\in\L}$ is an orthogonal basis of $L^2(\O)$, 
$\sum_{\y\in\L'} |\hat\f(\x+\y)|^2$ must be constant in order for the inner product to vanish 
when $z\neq y$. By the condition for $z=y=0$, we also learn that the constant  must be
$\mathrm{vol}(\O)^{-1}=\omega_0$. Thus we have 

\begin{lem}\label{lem-ONC}
Let $\f\in\mathcal{S}(\re^d)$. Then $\{\f(\cdot-z)\}_{z\in\L}$ is an 
orthonormal system if and only if 
\begin{equation}\label{eq-ons}
\sum_{\y\in\L'}|\hat\f(\x+\y)|^2 =\omega_0, \quad \x\in\re^d. 
\end{equation}
\end{lem}
\begin{ass}\label{ass-ONS}
$\f\in\mathcal{S}(\re^d)$ satisfies the condition \eqref{eq-ons}. Moreover, $\hat\f\in C_0^\infty(\re^d)$, 
and $|\hat\f(\x)|^2=\omega_0$ in a neighborhood of $0$.
%$\supp[\hat\f]\cap\L'=\{0\}$. 
\end{ass}

\begin{rem}
We can easily show functions satisfying Assumption~A exist. 
We sketch a  typical construction: Let $\O_0$ be the first Brillouin zone:  
\[
\O_0=\bigset{x\in\re^d}{|x|\leq \mathrm{dist}(x,y) \text{ for all }y\in\L'\setminus \{0\}}.
\]
$\O_0$ is a fundamental domain for $\L'$, and $\sqrt{\omega_0}$ times the characteristic function of 
$\O_0$, $\chi_{\O_0}(\x)$, satisfies the conditions except for the smoothness. 
Let $\g_\e*$, $\e>0$, be a mollifier and let $\hat \f_0=\g_\e*\chi_{\O_0}$
with small $\e>0$ so that $\supp[\hat\f_0]\cap\L' =\{0\}$. Then 
\[
\hat\f(\x) =\sqrt{\omega_0}\biggpare{\sum_{\y\in\L'} |\hat\f_0(\x+\y)|^2}^{-1/2}\hat\f_0(\x)
\]
satisfies Assumption~A, and $\hat\f$ is supported in an arbitrarily close neighborhood of $\O_0$. 
\end{rem}

%%%
We can now define the identification operator $J_h$. 
Suppose $\f$ satisfies Assumption~A and we set
\[
\f_{h;z}(x)=\f(h^{-1}(x-z)), \quad x\in\re^d, z\in h\L.
\]
Then we define 
\[
J_h v(x) = \sqrt{\omega_0}\sum_{z\in h\L} \f_{h;z}(x) v(z)
\]
for $v\in\ell^2(h\L)$.
Assumption~\ref{ass-ONS} implies it is an isometry from $\ell^2(h\L)$ to $L^2(\re^d)$. 
The adjoint operator is given by 
\[
J_h^* u(z) = \frac{1}{\sqrt{\omega_0}h^d} \int_{\re^d} \overline{\f_{h;z}(x)}u(x) dx, \quad z\in h\L, 
\]
for $u\in L^2(\re^d)$.

%%%%%

\subsection{Proof of Theorems~\ref{thm-free-hamiltonian} and \ref{thm-hamiltonian-with-potential}}
\label{subsec-lattice-proof}
At first, we recall $H_0$ is a Fourier multiplier with the symbol $p_0(2\pi\x)$, i.e., 
\[
\fourier H_0 \fourier^* =p_0(2\pi\x)\cdot\quad \text{for }\x\in \re^d. 
\]
Similarly, the symbol of $H_{0,h}$ is given by 
\[
F_h H_{0,h} F_h^* = p_{0,h}(2\pi\z)\cdot = 2h^{-2} \sum_{j=1}^K \m_j (1-\cos(h f_j\cdot2\pi \z))\cdot
\quad \text{for }\z\in h^{-1}\L'. 
\]
We denote
\[
T_h = \fourier J_h F_h^* \ :\ L^2(h^{-1}\O) \to L^2(\re^d).
\]
%\[
%Q_h = F_h J_h \fourier^* \ :\ L^2(\re^d) \to L^2(h^{-1}\O).
%\]

\begin{lem}
For $g\in L^2(h^{-1}\O)$, 
\[
T_h g(\x) =\omega_0^{-1/2} \hat\f(h\x)\tilde g(\x), \quad \x\in\re^d, 
\]
where $\tilde g$ is the periodic extension of $g$,
and for $f\in\mathcal{S}(\re^d)$, 
\[
T_h^* f(\x) = \omega_0^{-1/2}\sum_{\y\in h^{-1}\L'} \overline{\hat\f(h(\x+\y))}f(\x+\y), 
\quad \x\in h^{-1}\O.
\]
\end{lem}

\begin{proof}
We have 
\begin{align*}
T_h g(\x) &= \int_{\re^d} e^{-2\pi ix\cdot\x}\sqrt{\omega_0} \sum_{z\in h\L} \f(h^{-1}(x-z)) 
F_h^*g(z) dx\\
&= \sqrt{\omega_0}\sum_{z\in h\L} e^{-2\pi i z\cdot\x}  F_h^*g(z)
\int_{\re^d}e^{-2\pi i(x-z)\cdot \x}\f(h^{-1}(x-z))dx\\
&= \sqrt{\omega_0} h^d\sum_{z\in h\L} e^{-2\pi i z\cdot\x}  F_h^*g(z) \hat\f(h\x)
=\omega_0^{-1/2} \hat\f(h\x)\tilde g(\x),
\end{align*}
where we have used the inversion formula at the last line, and we identify $g\in L^2(h^{-1}\O)$
with a $(h^{-1}\L')$-periodic function $\tilde g$ on $\re^d$. 

Then we compute 
\begin{align*}
\jap{g,T_h^* f}&= 
\jap{T_h g, f}
=\int_{\re^d} \overline{T_h g(\x)}f(\x) d\x
=\omega_0^{-1/2} \int_{\re^d} \overline{\hat\f(h\x)\tilde g(\x)}f(\x) d\x\\
&=\omega_0^{-1/2} \sum_{\y\in h^{-1}\L'} \int_{h^{-1}\O}
\overline{\hat\f(h(\x+\y)) g(\x)}f(\x+\y) d\x\\
&=\int_{h^{-1}\O} \overline{g(\x)} \biggpare{\omega_0^{-1/2} \sum_{\y\in h^{-1}\L'} 
\overline{\hat\f(h(\x+\y))}f(\x+\y)} d\x,
\end{align*}
and this implies the assertion. 
\end{proof}

\begin{lem}\label{lem-1-JJ}
For $\m\in\co\setminus\re$, there is $C>0$ such that
\[
\bignorm{(1-J_hJ_h^*)(H_0-\m)^{-1}}_{\mathcal{B}(L^2(\re^d))}\leq Ch^2.
\]
\end{lem}

\begin{proof}
The proof is analogous to \cite{NaTad} Lemma~2.2, and we sketch the proof. We first note 
\[
\bignorm{(1-J_hJ_h^*)(H_0-\m)^{-1}}
=\norm{(1-T_h T_h^*)g_\m(\cdot)},
\]
where we denote $g_\m(\x)=(p_0(2\pi\x)-\m)^{-1}$. Then, using the above lemma, we have 
\begin{align*}
(1-T_hT_h^*)g_\m(\x) &= (1-\omega_0^{-1}|\hat\f(h\x)|^2)g_\m(\x)\\
&\quad -\omega_0^{-1}\hat\f(h\x)\sum_{n\neq 0, n\in\L'} \overline{\hat\f(h\x+n)}g_\m(\x+h^{-1}n)
\end{align*}
Thanks to Assumption~\ref{ass-ONS}, $(1-\omega_0^{-1}|\hat\f(h\x)|^2)g_\m(\x)$  
is supported in $|h\x|\geq \d$ with some $\d>0$. Since $|g_\m(\x)|\leq Ch^{2}$ on 
$\{\x\mid |\x|\geq \d h\}$ with some $C>0$, we learn 
\[
\sup_{\x\in\re^d} \bigabs{(1-\omega_0^{-1}|\hat\f(h\x)|^2)g_\m(\x)}\leq Ch^2.
\]
The same argument works for the other terms since the sum is finite and each term 
is supported in $|h\x+n|\geq \d$. These imply the assertion. 
\end{proof}

\begin{proof}[Proof of Theorem~\ref{thm-free-hamiltonian}]
The argument is essentially the same as the proof of \cite{NaTad} Lemma~2.3. 
We mainly describe the differences here. Naturally the sum over $n\in\ze^d$ in \cite{NaTad}
is replaced by the sum over $n\in\L'$. The symbol of $H_{0,h}$, i.e., the multiplier of 
$F_h H_{0,h} F_h^*$ is given by 
\[
p_{0,h}(2\pi\x)=2h^{-2} \sum_{j=1}^K  \m_j(1-\cos(h f_j\cdot2\pi\x)),
\]
and by the Taylor expansion, we have 
\[
p_{0,h}(2\pi\x) = \sum_{j=1}^K \m_j (f_j\cdot2\pi\x)^2 + O(h^2|\x|^4)=p_0(2\pi\x) +O(h^2|\x|^4). 
\]
Then we can carry out the same argument as in the proof of \cite{NaTad} Lemma~2.3. 
\end{proof}

As well as \cite{NaTad} Lemmas~2.4, 2.5, we have the following lemma. 
The proof is essentially the same, and we omit it.

\begin{lem}\label{lem-rel-bdd}
Under the Assumption~\ref{ass-V}, 
$V$ is $H$-bounded, and hence $H_0$ is also $H$-bounded.
Moreover $V_h$ is $H_h$-bounded uniformly in $h>0$, 
and hence $H_{0,h}$ is also $H_h$-bounded uniformly in $h>0$. 
\end{lem}

The next lemma is almost identical to \cite{NaTad} Lemma~2.6. 

\begin{lem}
Suppose $G$ is a bounded uniformly continuous function on $\re^d$. Then 
\[
\bignorm{GJ_h^* - J_h^* G}_{\mathcal{B}(L^2(\re^d),\ell^2(h\L))} \to 0
\quad \text{as }h\to 0.
\]
\end{lem}

Now Theorem~\ref{thm-hamiltonian-with-potential} and 
Corollary~\ref{cor-spectrum-lattice} are proved by the same argument as
that of Subsection~2.3 and Appendix~A in \cite{NaTad}, respectively.

%%%%
\subsection{Examples}
Here we discuss several typical examples. 

\begin{example}[Triangular lattice]\label{example-tri}
The lattice of the triangular lattice in $\re^2$ is generated by the pair of vectors 
$e_1=(1,0)$ and $e_2=(1/2,\sqrt{3}/2)$, i.e., $\L^{\mathrm{tr}}=\ze e_1+\ze e_2$. 
The graph of the triangular lattice has degree 6, and it is generated by $e_1$, $e_2$ and $e_3=(-1/2,\sqrt{3}/2)$, i.e., 
\[
\edges = ([0,e_1]+\L^{\mathrm{tr}})\cup ([0,e_2]+\L^{\mathrm{tr}})\cup ([0,e_3]+\L^{\mathrm{tr}})
\]
(as an undirected graph). Then the symbol of the 
(scaled) Laplace operator is given by 
\[
p_{0,h}(2\pi\x) = 2h^{-2} \sum_{j=1}^3 (1-\cos(he_j\cdot 2\pi\x)), 
\]
where we set $\m_j=1$ for all $j$. The symbol of the limit operator is 
\[
p_0(\x)=\sum_{j=1}^3 (e_j\cdot\x)^2 = \x_1^2 +(\x_1/2+\x_2\sqrt{3}/2)^2
+(-\x_1/2+\x_2\sqrt{3}/2)^2 = (3/2)|\x|^2, 
\]
i.e., the limit operator is $-(3/2)\triangle$. 
\end{example}

\begin{example}[Tetrahedral lattice]
The tetrahedral lattice in $\re^3$ is generated by the three vectors
\[
e_1=(1,0,0), \quad e_2=(1/2,\sqrt{3}/2,0), \quad e_3=(1/2,\sqrt{3}/6, \sqrt{2/3}).
\]
The graph has degree 12, and the set of edges is generated by six vectors: 
$e_1$, $e_2$, $e_3$, $(-1/2,\sqrt{3}/2,0)$ and
\[
(-1/2, \sqrt{3}/6, \sqrt{2/3}), \quad 
(0,- 1/\sqrt{3}, \sqrt{2/3}).
\]
If $\m_j=1$ for all $j$, we have the limit operator $-2\triangle$. 
\end{example}

\begin{example}[Octahedral lattice]

The lattice of the octahedral lattice in $\re^3$ is generated by the three vectors
\[
e_1=(1,0,0), \quad e_2=(0,1,0), \quad e_3=(1/2,1/2, 1/\sqrt{2}).
\]
The graph has degree 12, and the set of edges is generated by six vectors: 
$e_1$, $e_2$,  and $(\pm 1/2,\pm 1/2, 1/\sqrt{2})$. 
If $\m_j=1$ for all $j$, then the limit operator is again $-2\triangle$. 
\end{example}

%%%%%%%%%%

\section{Hexagonal lattice}\label{sec-hexagonal}
\subsection{Identification operator for the hexagonal lattice}

In the following sections, we use the notation in Subsection~\ref{subsec-intro-hex}. 
Let $\Lambda^{\mathrm{tr}}_h$ be the triangular lattice generated by $e_1$ and $e_2$, i.e., 
\[
\Lambda^{\mathrm{tr}}_h = \bigset{n_1h e_1+n_2h e_2}{n_1,n_2\in\ze}, 
\]
where $e_j$ are introduced in Subsection~\ref{subsec-intro-hex}. 
We note $\Lambda^{\mathrm{tr}}_h \supset \Lambda^{\mathrm{hex}}_h$, and this is the same lattice 
considered in Example~\ref{example-tri}. 
The norm of $\ell^2(\Lambda^{\mathrm{tr}}_h)$ is given by 
\[
\norm{u}_{\ell^2(\Lambda^{\mathrm{tr}}_h)}^2 = \omega_0 h^2\sum_{z\in\Lambda^{\mathrm{tr}}_h} |u(z)|^2, 
\quad u\in \ell^2(\Lambda^{\mathrm{tr}}_h), 
\]
where $\omega_0=\sqrt{3}/2$, and then the natural norm of $\ell^2(\Lambda^{\mathrm{hex}}_h)$ is given by 
\[
\norm{u}_{\ell^2(\Lambda^{\mathrm{hex}}_h)}^2 = \frac{3}{2} \omega_0 h^2\sum_{z\in\Lambda^{\mathrm{hex}}_h} |u(z)|^2, 
\quad u\in \ell^2(\Lambda^{\mathrm{hex}}_h). 
\]
We set an interpolation operator $\Theta_h : \ell^2(\Lambda^{\mathrm{hex}}_h)\to \ell^2(\Lambda^{\mathrm{tr}}_h)$ by 
\[
\Theta_h u(z) = \begin{cases}u(z) \quad&\text{if }z\in \Lambda^{\mathrm{hex}}_h, \\
{\displaystyle \frac{1}{6}\sum_{|z'-z|=h} u(z')} \quad &\text{if }z\notin \Lambda^{\mathrm{hex}}_h
\end{cases}
\]
for $u\in \ell^2(\Lambda^{\mathrm{hex}}_h)$. 
Let $\f(x)\in\mathcal{S}(\re^2)$ be such that 
$\{\f_{h;z}(x)=\f(h^{-1}(x-z))\}_{z\in\Lambda^{\mathrm{tr}}_h}$
is an orthonormal system of $L^2(\re^2)$, $\hat\f\in C_0^\infty(\re^2)$ and that $|\hat\f(\x)|^2=\omega_0$ in a neighborhood of $0$, 
as in Assumption~\ref{ass-ONS}.  We denote the identification operator for the 
triangular lattice by
$J^{\mathrm{tr}}_h : \ell^2(\Lambda^{\mathrm{tr}}_h) \to L^2(\re^2)$, i.e., 
\[
J^{\mathrm{tr}}_h v(x) = \sqrt{\omega_0}\sum_{z\in\Lambda^{\mathrm{tr}}_h} \f_{h;z}(x) v(z), \quad x\in\re^2, 
\]
for $v\in \ell^2(\Lambda^{\mathrm{tr}}_h)$.  
%We also recall 
%\[
%(J^{\mathrm{tr}}_h)^* u(z) = \frac{1}{\sqrt{\omega_0}h^{d}}\int\overline{\f_{h;z}(x)}u(x) dx, \quad z\in \Lambda^{\mathrm{tr}}_h, 
%\]
%for $u\in L^2(\re^2)$.
We now define $J^{\mathrm{hex}}_h$ by 
\[
J^{\mathrm{hex}}_h = J^{\mathrm{tr}}_h \Theta_h\ : \ \ell^2(\Lambda^{\mathrm{hex}}_h)
\to L^2(\re^2). 
\]

We note, by Theorem~1.1 and Example~5.1, we have 
\begin{equation}\label{eq-hex-keystep1}
\bignorm{(H^{\mathrm{tr}}_0-\m)^{-1}-J^{\mathrm{tr}}_h (H^{\mathrm{tr}}_h-\m)^{-1}
(J^{\mathrm{tr}}_h)^*}_{\mathcal{B}(L^2(\re^2))} \leq Ch^2,
\end{equation}
where $H^{\mathrm{tr}}_0=-(3/2)\triangle$ on $L^2(\re^2)$ and 
$H^{\mathrm{tr}}_h$ is the Hamiltonian on $\ell^2(\L^{\mathrm{tr}}_h)$. 
We will show (see Lemma~\ref{lem-hex-key2}): 
\begin{equation}\label{eq-hex-keystep2}
\bignorm{(\tfrac12 H^{\mathrm{tr}}_h-\m)^{-1}-\Theta_h (H^{\mathrm{hex}}_h-\m)^{-1}
\Theta_h^*}_{\mathcal{B}(\ell^2(\L^{\mathrm{tr}}_h))} \leq Ch. 
\end{equation}
It is easy to see \eqref{eq-hex-keystep1} and \eqref{eq-hex-keystep2} imply the first estimate of Theorem~\ref{thm-hexa} since 
$H^{\mathrm{hex}}_0=\frac12 H^{\mathrm{tr}}_0$. 

%%%%%%%%%%%%%%%%%
\subsection{Operators $H^{\mathrm{hex}}_h$, $H^{\mathrm{tr}}_h$ and $\Theta_h$ in the Fourier space}

We recall the basis of $\Gamma_h$ is $hf_1$ and $hf_2$, where $f_1= (3/2,\sqrt{3}/2)$ and 
$f_2=(3/2,\sqrt{3}/2)$. The dual basis for $(f_1,f_2)$ is given by 
\[
f_1'=(1/3,\sqrt{3}/3)=(2/3)e_1, \quad f_2'=(1/3,-\sqrt{3}/3)=(2/3)e_2,
\]
and the then the dual lattice for $\Gamma_h$ is given by 
$\Gamma_h'=h^{-1}f_1'\ze +h^{-1}f_2'\ze$.
On the other hand, the dual basis for $(e_1,e_2)$ is given by
\[
e_1'=(1, \sqrt{3}/3) = (2/3)f_1, \quad e_2'=(1, -\sqrt{3}/3) = (2/3)f_2,
\]
and the dual lattice is $(\L^{\mathrm{tr}}_h)'=h^{-1}e_1'\ze+h^{-1}e_2'\ze$. We note 
\[
e_1'=2f_1'+f_2', \quad e_2'= f_1'+2f_2', 
\]
and in particular, $(\L^{\mathrm{tr}}_h)'$ is a submodule of $\Gamma_h'$, and 
$\#(\Gamma_h'/(\L^{\mathrm{tr}}_h)')=3$. 

%%% Figure 2 removed %%%

We identify $u\in\ell^2(\Lambda^{\mathrm{tr}}_h)$ (resp. $v\in\ell^2(\Lambda^{\mathrm{hex}}_h)$) 
with 
\[
\begin{pmatrix} u(\cdot) \\ u(\cdot+he_3) \\ u(\cdot-he_2) \end{pmatrix} \in\ell^2(\Gamma_h; \co^3)
 \quad 
\text{(resp. }\begin{pmatrix} v(\cdot) \\ v(\cdot+he_3) \end{pmatrix} \in \ell^2(\Gamma_h; \co^2)\text{).}
\]
These identifications are unitary with the norms
\begin{align}\label{norms}
\norm{v}_{\ell^2(\Gamma_h;\co^n)}^2 = \frac{3\omega_0 h^2}{n} \sum_{z\in\Gamma_h} |v(z)|^2
\end{align}
for $v\in\ell^2(\Gamma_h;\co^n)$, $n=2,3$. We note $|\det(f_1,f_2)|=3\sqrt{3}/2=3\omega_0$. 

For $v\in\ell^2(\L_h;\co^n)$, we denote the Fourier transform for $\Gamma_h$ as 
\[
F_h v(\x)=3\omega_0 h^2 \sum_{y\in \Gamma_h} e^{-2\pi i y\cdot\x}v(y), 
\quad \x\in\O_h, 
\]
where $\O_h$ is a fundamental domain of $\Gamma_h'$. 
%, and the inverse is given by 
%\[
%F_h^* w(x) =\int_{\O_h} e^{2\pi i x\cdot\x} w(\x)d\x, \quad x\in \Gamma_h,
%\]
%for $w\in L^2(\O_h;\co^n)$. 
$F_h$ is a unitary from $\ell^2(\Gamma_h;\co^n)$ to $L^2(\O_h;\co^n)$ with the norm
\[
\norm{w}_{L^2(\O_h;\co^n)}^2 = \frac1n \int_{\O_h}|w(\x)|^2 d\x.
\] 
In general, we consider $\co^n$ as a Hilbert space with the inner product
\begin{equation}\label{eq-co-ip}
\jap{u,v}_{\co^n} =\frac1n \sum_{j=1}^n \overline{u_j} v_j, \quad u,v\in\co^n.
\end{equation}

Now we write down the representations of operators $H^{\mathrm{hex}}_h$, 
$H^{\mathrm{tr}}_h$ and $\Theta_h$ in the Fourier space $L^2(\O_h,\co^n)$. 
For simplicity, we denote 
\[
\phi_f(\x):=e^{2\pi ih f\cdot\x}, \quad \x\in\re^2
\]
for $f\in\re^2$. The multiplication by $\phi_f$ on $\O_h$ corresponds to the shift by $h f$ in 
$\L^{\mathrm{tr}}_h$, etc., when $f\in \L^{\mathrm{tr}}_1$. 

Now if we set
\[
\hat H^{\mathrm{tr}}_h(\x)
= h^{-2}{\small \begin{pmatrix}
6 & -1 - \phi_{f_1} - \phi_{f_2} & -1 - \phi_{f_2} - \phi_{f_2-f_1} \\
-1-\phi_{-f_1}-\phi_{-f_2} & 6 & -1-\phi_{-f_1}-\phi_{f_2-f_1} \\
-1 - \phi_{-f_2} - \phi_{f_1-f_2} & -1-\phi_{f_1}-\phi_{f_1-f_2} & 6
\end{pmatrix}}, 
\]
then it is easy to confirm $F_h H^{\mathrm{tr}}_h F_h^* =\hat H^{\mathrm{tr}}_h(\x)\cdot$ 
on $L^2(\O_h;\co^3)$. Similarly, if we set
\[
\hat H^{\mathrm{hex}}_h(\x)
= h^{-2}\begin{pmatrix}
3 & -1 - \phi_{f_1} - \phi_{f_2} \\
-1-\phi_{-f_1}-\phi_{-f_2} & 3 
\end{pmatrix}, 
\]
then one can show $F_h H^{\mathrm{hex}}_h F_h^* =\hat H^{\mathrm{hex}}_h(\x)\cdot$ 
on $L^2(\O_h;\co^2)$. We also set 
\[
\hat\Theta_h(\x) =
\begin{pmatrix} 1 & 0 \\ 0 & 1 \\ \frac16(1+\phi_{-f_2}+\phi_{f_1-f_2}) 
& \frac16(1+\phi_{f_1}+\phi_{f_1-f_2}) \end{pmatrix},
\]
and then we have $F_h \Theta_h F_h^* =\hat \Theta_h(\x)\cdot$ as an operator from 
$L^2(\O_h;\co^2)$ to $L^2(\O_h;\co^3)$. This also implies $F_h \Theta_h^* F_h^*$ is 
represented by the multiplication operator by
\[
\hat \Theta_h^*(\x) 
=\frac23 \begin{pmatrix} 1 & 0 & \frac16(1+\phi_{f_2}+\phi_{-f_1+f_2}) 
\\ 0 & 1 & \frac16(1+\phi_{-f_1}+\phi_{-f_1+f_2}) \end{pmatrix} 
\]
as an operator from $L^2(\O_h;\co^3)$ to $L^2(\O_h;\co^2)$. We note the multiplication factor 
2/3 appears because of the different definitions of inner products (see \eqref{eq-co-ip}). 

It is well-known that the spectrum of $H^{\mathrm{hex}}_h$ is computed by the eigenvalues 
of $\hat H^{\mathrm{hex}}_h(\x)$ as follows. 

\begin{lem}\label{lem-spec-hex}
Let $\alpha(\x)=1+\phi_{f_1}(\x)+\phi_{f_2}(\x)$. Then the eigenvalues and the eigenvectors of 
$\hat H^{\mathrm{hex}}_h(\x)$ are given by 
\begin{align*}
E^{\mathrm{hex}}_\pm(\x)&=h^{-2}(3\pm|\a(\x)|), \\
w^{\mathrm{hex}}_\pm(\x)&= \begin{pmatrix} 1 \\ \mp \overline{\a(\x)}/|\a(\x)| \end{pmatrix}, 
\end{align*}
respectively. In particular, the spectrum of $H^{\mathrm{hex}}_h$ is $[0, 6h^{-2}]$ 
and it is absolutely continuous. 
\end{lem}

The proof is elementary, and we omit it. 

The spectrum of $H^{\mathrm{tr}}_h$ is similarly straightforward, but it looks slightly 
complicated in our representation. If we consider $H^{\mathrm{tr}}_h$ as an operator 
on $\L^{\mathrm{tr}}_h$, then the (generalized) eigenvalues are computed using the 
Fourier transform and they are given by
\[
E^{\mathrm{tr}}(\x)= h^{-2}\biggpare{6-2\sum_{j=1}^3\cos(2\pi h e_j\cdot\x)}, 
\]
where $\x$ is an element of a fundamental domain for $(\L^{\mathrm{tr}}_h)'$ 
(or $\x\in \re^2/ (\L^{\mathrm{tr}}_h)'$), and the (generalized) 
eigenfunctions are given by 
\[
\psi_\x(z)=\exp(2\pi i \x\cdot z), \quad z\in\L^{\mathrm{tr}}_h
\]
(see Example~\ref{example-tri}). 

\begin{lem}\label{lem-spec-tr}
Let $\x\in\O_h$ and $E^{\mathrm{tr}}(\cdot)$  as above. Then the eigenvalues and the eigenvectors of 
$\hat H^{\mathrm{tr}}_h(\x)$ are given by 
\[
E^{\mathrm{tr}}_0(\x)=E^{\mathrm{tr}}(\x); \quad 
E^{\mathrm{tr}}_\pm(\x) = E^{\mathrm{tr}}(\x\pm h^{-1}f_1'),
\]
and 
\[
w^{\mathrm{tr}}_0(\x)= \begin{pmatrix} 1\\ \phi_{e_3}(\x)\\ \phi_{-e_2}(\x) \end{pmatrix};  
\quad 
w^{\mathrm{tr}}_\pm(\x)= \begin{pmatrix} 1\\ \phi_{e_3}(\x\pm h^{-1}f_1')
\\ \phi_{-e_2}(\x\pm h^{-1}f_1') \end{pmatrix}, 
\]
respectively. 
\end{lem}

\begin{proof}
By our identification, the generalized eigenfunction $\g_\x$ is transformed to 
\[
\Psi_\x(z) = e^{2\pi i\x\cdot z}\begin{pmatrix}1 \\ \phi_{e_3}(\x) \\ \phi_{-e_2}(\x) \end{pmatrix}
\in \ell^\infty(\Gamma_h;\co^3), 
\quad z\in\Gamma_h. 
\]
Now we note $e^{2\pi i\x\cdot z}=e^{2\pi i\x'\cdot z}$ for any $z\in\Gamma_h$ 
if and only if $\x-\x'\in \Gamma_h'$. Since $\#(\Gamma_h'/(\L^{\mathrm{tr}}_h)')=3$, 
for each $\x_0\in\O_h$ there are 3 such elements in $\re^2/(\L^{\mathrm{tr}}_h)'$, 
and we can choose $\x_0$ and $\x_0\pm h^{-1}f_1'$ as the representative elements. 
Thus we have 3 eigenfunctions: 
\[
e^{2\pi i\x_0\cdot z}w^{\mathrm{tr}}_0(\x_0), 
\quad 
e^{2\pi i\x_0\cdot z}w^{\mathrm{tr}}_\pm(\x_0)
\]
with the eigenvalues $E^{\mathrm{tr}}_0(\x_0)$ and $E^{\mathrm{tr}}_\pm(\x_0)$ 
for each $\x_0\in\O_h$. By applying the Fourier transform, we learn 
$\d(\x-\x_0)w^{\mathrm{tr}}_0(\x_0)$ 
and $\d(\x-\x_0)w^{\mathrm{tr}}_\pm(\x_0)$
are (generalized) eigenfunctions of the multiplication operator $\hat H^{\mathrm{tr}}_h(\x)\cdot$ 
on $\re^2/\Gamma_h'$. Thus we have 
\begin{align*}
\d(\x-\x_0) \hat H^{\mathrm{tr}}_h(\x_0) w^{\mathrm{tr}}_\#(\x_0)
&=\hat H^{\mathrm{tr}}_h(\x) \d(\x-\x_0)w^{\mathrm{tr}}_\#(\x_0)\\
&= E^{\mathrm{tr}}_\#(\x_0) \d(\x-\x_0)w^{\mathrm{tr}}_\#(\x_0)\\
&= \d(\x-\x_0) E^{\mathrm{tr}}_\#(\x_0) w^{\mathrm{tr}}_\#(\x_0),
\end{align*}
i.e., $\hat H^{\mathrm{tr}}_h(\x_0) w^{\mathrm{tr}}_\#(\x_0)
= E^{\mathrm{tr}}_\#(\x_0) w^{\mathrm{tr}}_\#(\x_0)$, where 
$\#=0$ or $\pm$. These complete the proof. 
\end{proof}

\begin{rem}
In fact, by direct computations, we have 
\begin{align*}
\phi_{e_3}(\x\pm h^{-1}f_1')&= e^{\pm 2\pi ie_3\cdot f_1'}\phi_{e_3}(\x) 
= -e^{\mp\pi i/3}\phi_{e_3}(\x), \\
\phi_{-e_2}(\x\pm h^{-1}f_1')&= e^{\mp 2\pi ie_2\cdot f_1'}\phi_{-e_2}(\x) 
= -e^{\pm \pi i/3}\phi_{-e_2}(\x). 
\end{align*}
Using these expressions, we can easily confirm 
$w^{\mathrm{tr}}_0$ and $w^{\mathrm{tr}}_\pm$ form an orthonormal 
basis of $\co^3$ for each $\x$, which is expected from the general theory. 
\end{rem}

%%%%%%%%%%%%%%%%%
\subsection{Proof of Theorem~\ref{thm-hexa}}

Now we fix a fundamental domain for $\Gamma_h'$.  
We let $\O_h$ be the first Brillouin zone:
\[
\O_h := \bigset{\x\in\re^2}{ |x|\leq |\x-f'| \text{ for any }f'\in \Gamma'_h} .
\]

%%% Figure 3 removed %%%

We denote the projection to $w^{\mathrm{tr}}_0(\x)$ in $\co^3$ by $\Pi^{\mathrm{tr}}(\x)$: 
\[
\Pi^{\mathrm{tr}}(\x) u = \jap{w^{\mathrm{tr}}_0(\x),u}_{\co^3} w^{\mathrm{tr}}_0(\x), 
\quad u\in \co^3, 
\]
and the projection to $w^{\mathrm{hex}}_-(\x)$ in $\co^2$ by $\Pi^{\mathrm{hex}}(\x)$: 
\[
\Pi^{\mathrm{hex}}(\x)v = \jap{w^{\mathrm{hex}}_-(\x),v}_{\co^2} w^{\mathrm{hex}}_-(\x), 
\quad v\in \co^2. 
\]
Then we have 

\begin{lem}\label{lem-tr-proj}
Let $\m\in\co\setminus\re$. There is $C>0$ such that 
\[
\bignorm{(\hat H^{\mathrm{tr}}_h(\x) -\m)^{-1}- (E^{\mathrm{tr}}_0(\x) -\m)^{-1}\Pi^{\mathrm{tr}}(\x)}_{\co^3\to\co^3}
\leq C h^2
\]
and
\[
\bignorm{(\hat H^{\mathrm{hex}}_h(\x) -\m)^{-1}- (E^{\mathrm{hex}}_-(\x) -\m)^{-1}\Pi^{\mathrm{hex}}(\x)}_{\co^2\to\co^2}
\leq C h^2.
\]
\end{lem}

\begin{proof}
We note $h^2 E^{\mathrm{tr}}(\x)$ vanishes only at 0 in $\re^2/(\L^{\mathrm{tr}}_h)'$ and 
the neighborhood of 0 is contained in $\O_h$ under our setting. Hence $h^2 E^{\mathrm{tr}}_\pm(\x)$ 
do not vanish on $\O_h$, and bounded from below by some $c>0$. 
The first estimate follows from this observation. The second estiate follows from the fact 
$E^{\mathrm{hex}}_+(\x)\geq 3h^{-2}$ for all $\x$. 
\end{proof}

The following elementary estimate follows from the fact: $|\x|\leq Ch^{-1}$ in $\O_h$, and useful 
in the following. 

\begin{lem} \label{lem-hex-ellipticity}
Let $\m\in\co\setminus \re$. Then there is $C>0$ such that 
\[
\bigabs{(E^{\mathrm{tr}}_0(\x) -\m)^{-1}}\leq C\jap{\x}^{-2}, \quad 
\bigabs{(E^{\mathrm{hex}}_-(\x) -\m)^{-1}}
\leq C \jap{\x}^{-2}
\]
for all $\x\in\O_h$.
\end{lem}

We now consider the operators $\hat\Theta_h(\x)$, $\hat\Theta_h(\x)^*$, 
$w^{\mathrm{tr}}_0(\x)$ and $w^{\mathrm{hex}}_-(\x)$ near $\x=0$. 
By the Taylor expansions, we learn 
\[
\hat\Theta_h(\x) =
\begin{pmatrix} 1 & 0 \\ 0 & 1 \\ \frac12 & \frac12 
\end{pmatrix} 
+\frac{2\pi i h}{2} \begin{pmatrix} 0 & 0 \\ 0 & 0 \\-e_2\cdot\x &  e_1\cdot\x
\end{pmatrix} 
+O(h^2|\x|^2), \\
\]
\[
\hat \Theta_h^*(\x) 
=\frac23 \begin{pmatrix} 1 & 0 & \frac12 \\ 0 & 1 & \frac12 \end{pmatrix} 
+\frac{2\pi i h}{3} \begin{pmatrix} 0 & 0 & e_2\cdot\x \\ 0 & 0 &  -e_1\cdot\x \end{pmatrix} 
+O(h^2|\x|^2), 
\]
\[
w^{\mathrm{tr}}_0(\x)= \begin{pmatrix} 1\\ 1 \\ 1\end{pmatrix} 
+2\pi i h \begin{pmatrix} 0\\ e_3\cdot\x \\ -e_2\cdot\x\end{pmatrix} 
+O(h^2|\x|^2), 
\]
\[
w^{\mathrm{hex}}_-(\x)= \begin{pmatrix} 1 \\ 1 \end{pmatrix} 
+2\pi i h \begin{pmatrix} 0 \\ e_3\cdot \x \end{pmatrix} 
+O(h^2|\x|^2). 
\]
Combining these with the above lemmas, we can show the following basic estimates:

\begin{lem}\label{lem-hex-key1}
For $\m\in\co\setminus\re$, there is $C>0$ such that
\begin{align*}
\norm{(1-\Theta_h \Theta_h^*)(H^{\mathrm{tr}}_h/2-\m)^{-1}}_{\mathcal{B}(\ell^2(\Lambda^{\mathrm{tr}}_h))}&\leq Ch^2 , \\
\norm{(1-\Theta_h^* \Theta_h)(H^{\mathrm{hex}}_h-\m)^{-1}}_{\mathcal{B}(\ell^2(\Lambda^{\mathrm{hex}}_h))}&\leq Ch^2 .
\end{align*}
\end{lem}

\begin{proof}
For $\x\in\O_h$, using the above observations, we have 
\begin{align}
\hat\Theta_h(\x) w^{\textrm{hex}}_-(\x) 
&= \begin{pmatrix} 1 & 0 \\ 0 & 1 \\ \frac{1}{2} & \frac{1}{2} \end{pmatrix}
\begin{pmatrix} 1 \\ 1 \end{pmatrix} 
+2\pi ih\begin{pmatrix} 1 & 0 \\ 0 & 1 \\ \frac12 & \frac12 \end{pmatrix} 
\begin{pmatrix} 0 \\ e_3\cdot \x \end{pmatrix} \nonumber\\
& +\frac{2\pi i h}{2} \begin{pmatrix} 0 & 0 \\ 0 & 0 \\-e_2\cdot\x &  e_1\cdot\x\end{pmatrix} 
\begin{pmatrix} 1 \\ 1 \end{pmatrix} +O(h^2|\x|^2) \nonumber \\
&= \begin{pmatrix} 1 \\ 1 \\ 1 \end{pmatrix}
+2\pi i h \begin{pmatrix} 0\\ e_3\cdot\x \\ -e_2\cdot\x\end{pmatrix}  +O(h^2|\x|^2) \nonumber\\
&= w^{\textrm{tr}}_0(\x) + O(h^2|\x|^2) , \label{eq-ThwZ}
\end{align}
\begin{align}
\hat\Theta_h^*(\x) w^{\textrm{tr}}_0(\x) 
&= \frac23 \begin{pmatrix} 1 & 0 & \frac{1}{2} \\ 0 & 1 & \frac{1}{2} \end{pmatrix} 
\begin{pmatrix} 1 \\ 1 \\ 1 \end{pmatrix}
+2\pi i h \frac23\begin{pmatrix} 1 & 0 & \frac{1}{2} \\ 0 & 1 & \frac{1}{2} \end{pmatrix} 
\begin{pmatrix} 0\\ e_3\cdot\x \\ -e_2\cdot\x\end{pmatrix}
\nonumber\\
& +\frac{2\pi i h}{3} \begin{pmatrix} 0 & 0 & e_2\cdot\x \\ 0 & 0 &  -e_1\cdot\x \end{pmatrix} 
\begin{pmatrix} 1 \\ 1 \\ 1\end{pmatrix} 
+O(h^2|\x|^2)\nonumber\\
&= \begin{pmatrix} 1 \\ 1 \end{pmatrix}
+2\pi i h \begin{pmatrix} 0 \\ e_3\cdot\x \end{pmatrix}+O(h^2|\x|^2)\nonumber \\
&= w^{\textrm{hex}}_-(\x) + O(h^2|\x|^2).\label{eq-ThSwZ}
\end{align}
In particular, we obtain 
\begin{equation}\label{eq-ThThS}
\hat\Theta_h(\x)\hat\Theta_h(\x)^* w^{\textrm{tr}}_0(\x) 
=w^{\textrm{tr}}_0(\x) +O(h^2|\x|^2),
\end{equation}
and 
\begin{equation}\label{eq-ThSTh}
\hat\Theta_h(\x)^*\hat\Theta_h(\x) w^{\textrm{hex}}_-(\x) 
=w^{\textrm{hex}}_-(\x) +O(h^2|\x|^2). 
\end{equation}
Thus we have, using \eqref{eq-ThThS}, Lemmas~\ref{lem-tr-proj} and \ref{lem-hex-ellipticity}, 
\begin{align*}
&\norm{(1-\Theta_h \Theta_h^*)(\hat H^{\mathrm{tr}}_h/2-\m)^{-1}}_{\mathcal{B}(\ell^2(\Lambda^{\mathrm{tr}}_h))} \\
&= 2\norm{(1-\hat\Theta_h \hat\Theta_h^*)(\hat H^{\mathrm{tr}}_h(\cdot)-2\m)^{-1}}_{\mathcal{B}(L^2(\O_h;\co^3))} \\
&= \frac23\sup_{\x\in\O_h} \bigabs{(E^{\mathrm{tr}}_0(\x)-2\mu)^{-1}(1-\hat\Theta_h(\x) \hat\Theta_h^*(\x)) \Pi^{\textrm{tr}}(\x)} + O(h^2) \\
&= \sup_{\x\in\O_h} \bigabs{O(\jap{\x}^{-2}) O(h^2|\x|^2)}+ O(h^2) 
 = O(h^2).
\end{align*}
The other estimate is similarly proved using \eqref{eq-ThSTh}, Lemmas~\ref{lem-tr-proj} 
and \ref{lem-hex-ellipticity}. 
\end{proof}

\begin{lem} \label{lem-hex-key2}
Let $\m\in\co\setminus\re$. Then there is $C>0$ such that 
\[
\norm{\Theta_h^* (H^{\mathrm{tr}}_h/2-\m)^{-1}- (H^{\mathrm{hex}}_h-\m)^{-1} 
\Theta_h^*}_{\mathcal{B}(\ell^2(\Lambda^{\mathrm{tr}}_h), \ell^2(\Lambda^{\mathrm{hex}}_h))}
\leq Ch^2. 
\]
\end{lem}

\begin{proof}
We first note 
\[
E^{\mathrm{tr}}_0(\x)= \frac32|2\pi\x|^2 + O(h^2|\x|^4)
\]
and 
\[
E^{\mathrm{hex}}_-(\x)= \frac34|2\pi\x|^2 +O(h^2|\x|^4)
\]
which are proved by elementary computations as well as in Subsection~\ref{subsec-lattice-proof}. We have, in particular, 
\begin{equation}\label{eq-EV-diff}
E^{\mathrm{hex}}_-(\x) - E^{\mathrm{tr}}_0(\x)/2 =O(h^2|\x|^4).
\end{equation}

By Lemma~\ref{lem-tr-proj}, we have 
\[
\hat\Theta_h^*(\hat H^{\mathrm{tr}}_h(\x)/2-\m)^{-1} 
=(E^{\mathrm{tr}}_0(\x)/2-\m)^{-1} \hat\Theta_h(\x)^*\Pi^{\mathrm{tr}}(\x) +O(h^2),
\]
and 
\[
(\hat H^{\mathrm{hex}}_h(\x)-\m)^{-1} \hat\Theta_h^*(\x) 
= (E^{\mathrm{hex}}_-(\x) -\m)^{-1} \Pi^{\mathrm{hex}}(\x) \hat\Theta_h^*(\x) +O(h^2). 
\]
If we denote 
\[
Bu = w^{\textrm{hex}}_-(\x)\jap{w^{\textrm{tr}}_0(\x),u}_{\co^3}, \quad u\in \co^3, 
\]
then 
\[
\hat\Theta_h(\x)^*\Pi^{\mathrm{tr}}(\x) =B+O(h^2|\x|^2), \quad 
 \Pi^{\mathrm{hex}}(\x) \hat\Theta_h^*(\x) =B+O(h^2|\x|^2). 
\]
Then, using Lemma~\ref{lem-hex-ellipticity}, and \eqref{eq-EV-diff}, we learn, 
\begin{align*}
&\hat\Theta_h^*(\hat H^{\mathrm{tr}}_h(\x)/2-\m)^{-1} - 
(\hat H^{\mathrm{hex}}_h(\x)-\m)^{-1} \hat\Theta_h^*(\x) \\
& \quad = \bigpare{(E^{\mathrm{tr}}_0(\x)/2-\m)^{-1} - (E^{\mathrm{hex}}_-(\x) -\m)^{-1}}
B +O(h^2)\\
&\quad =(E^{\mathrm{tr}}_0(\x)/2-\m)^{-1}(E^{\mathrm{hex}}_-(\x) - E^{\mathrm{tr}}_0(\x)/2)
(E^{\mathrm{hex}}_-(\x) -\m)^{-1}B +O(h^2)\\
&\quad =O(\jap{\x}^{-2})O(h^2|\x|^4)O(\jap{\x}^{-2}) +O(h^2) = O(h^2). 
\end{align*}
This is equivalent to the conclusion of the lemma. 
\end{proof}

Combining these lemmas, we can now conclude the convergence of the continuum limit 
for $H^{\mathrm{tr}}_h$ and $H^{\mathrm{hex}}_h$. 

\begin{lem}\label{lem-hex-key3}
Let $\m\in\co\setminus\re$. Then there is $C>0$ such that
\begin{equation*}
\bignorm{(H^{\mathrm{tr}}_h/2-\m)^{-1}-\Theta_h (H^{\mathrm{hex}}_h-\m)^{-1}
\Theta_h^*}_{\mathcal{B}(\ell^2(\L^{\mathrm{tr}}_h))} \leq Ch^2, 
\end{equation*}
and 
\begin{equation*}
\bignorm{\Theta_h^*(H^{\mathrm{tr}}_h/2-\m)^{-1}\Theta_h - (H^{\mathrm{hex}}_h-\m)^{-1}
}_{\mathcal{B}(\ell^2(\L^{\mathrm{hex}}_h))} \leq Ch^2. 
\end{equation*}
\end{lem}

\begin{proof}[Proof of Theorem~\ref{thm-hexa}] 
We recall $J^{\mathrm{hex}}_h = J^{\mathrm{tr}}_h \Theta_h$, and 
$H^\mathrm{hex}_0 =H^{\mathrm{tr}}_0/2$. Then we have 
\begin{align*}
&(H^{\mathrm{hex}}_0-\m)^{-1}-J^{\mathrm{hex}}_h (H^{\mathrm{hex}}_h-\m)^{-1} 
(J^{\mathrm{hex}}_h)^* \\
&\quad = (H^{\mathrm{tr}}_0/2-\m)^{-1}-J^{\mathrm{tr}}_h (H^{\mathrm{tr}}_h/2-\m)^{-1} 
(J^{\mathrm{tr}}_h)^* \\
&\qquad + J^{\mathrm{tr}}_h \bigpare{(H^{\mathrm{tr}}_h/2-\m)^{-1} -\Theta_h (H^{\mathrm{hex}}_h/2-\m)^{-1} \Theta_h^*}(J^{\mathrm{tr}}_h)^*.  
\end{align*}
Thus we conclude 
\begin{align*}
&\bignorm{(H^{\mathrm{hex}}_0-\m)^{-1}-J^{\mathrm{hex}}_h (H^{\mathrm{hex}}_h-\m)^{-1} 
(J^{\mathrm{hex}}_h)^*} \\
&\quad \leq  2\bignorm{(H^{\mathrm{tr}}_0-2\m)^{-1}-J^{\mathrm{tr}}_h (H^{\mathrm{tr}}_h-2\m)^{-1} 
(J^{\mathrm{tr}}_h)^*} \\
&\qquad + \bignorm{J^{\mathrm{tr}}_h}^2  \bignorm{(H^{\mathrm{tr}}_h/2-\m)^{-1} -\Theta_h (H^{\mathrm{hex}}_h-\m)^{-1} \Theta_h^*}\\
&\quad \leq Ch^2
\end{align*}
by Lemma~\ref{lem-hex-key3} and Theorem~\ref{thm-free-hamiltonian}. 
The other estimate is proved similarly. 
\end{proof}

%%%%%%%%%%%%%%%

\section{Elliptic operators: Proof of Theorem~\ref{thm-elliptic-convergence}}\label{sec-elliptic-proof}

General strategy of the proof of Theorem~\ref{thm-elliptic-convergence} is as follows. 
We use the same notation as in Section~\ref{sec-standard-lattices} with $\Lambda=h\ze^d$. 
As usual, we estimate 
\[
J_h^*(P-z)^{-1} - (P_h^\pm -z)^{-1} J_h^*
= (P_h^\pm -z)^{-1} (P_h^\pm J_h^* -J_h^*P)(P-z)^{-1}.
\]
We use the matrix notation: 
\[
\pmb{D}_h^\pm : u \in \ell^2(h\ze^d) \mapsto \begin{pmatrix} D_{h;1}^\pm u \\ \vdots \\ D_{h;d}^\pm u \end{pmatrix}
\in [\ell^2(h\ze^d)]^{\oplus d},
\]
and 
\[
\pmb{A}_h= (a_{j,k}(x))_{j,k=1}^d : [\ell^2(h\ze^d)]^{\oplus d} \to [\ell^2(h\ze^d)]^{\oplus d},
\]
so that 
\begin{equation}\label{eq-Ph-decomp}
P_h^\pm = (\pmb{D}_h^\pm)^* \pmb{A}_h \pmb{D}_h^\pm+V_h \quad \text{on }\ell^2(h\ze^d). 
\end{equation}

Analogously to the decomposition \eqref{eq-Ph-decomp}, we set
\[
\pmb{D} : u \in H^1(\re^d) \mapsto \begin{pmatrix} D_1 u \\ \vdots \\ D_d u\end{pmatrix}
\in [L^2(\re^d)]^{\oplus d},
\]
and 
\[
\pmb{A}= (a_{j,k}(x))_{j,k=1}^d : [L^2(\re^d)]^{\oplus d} \to [L^2(\re^d)]^{\oplus d},
\]
so that 
\begin{equation}\label{eq-P-decomp}
P= \pmb{D}^* \pmb{A} \pmb{D} +V \quad \text{on } L^2(\re^d)
\end{equation}
with the form domain $\mathcal{Q}(P)=\mathcal{D}(P^{1/2})=H^1(\re^d)\cap\mathcal{Q}(V)$. 
Then we consider 
\begin{align*}
P_h^\pm J_h^* -J_h^*P
&= (\pmb{D}_h^\pm)^* \pmb{A}_h \pmb{D}_h^\pm  J_h^* -J_h^* \pmb{D}^* \pmb{A} \pmb{D}  
+ V_h J_h^* - J_h^* V \\
&=(\pmb{D}_h^\pm)^* \pmb{A}_h (\pmb{D}_h^\pm  J_h^* -J_h^* \pmb{D}) 
+ (\pmb{D}_h^\pm)^* (\pmb{A}_h J_h^* -J_h^*  \pmb{A})  \pmb{D}\\
&\quad + (J_h \pmb{D}_h^\pm- \pmb{D} J_h)^*  \pmb{A} \pmb{D}
 + V_h J_h^* - J_h^* V
\end{align*}
where we denote $J_h^*\mathbf{1}_{d}=\mathrm{diag}(J_h^*)$ by the same symbol $J_h^*$ for simplicity.  
We will show
\begin{equation}\label{eq-elli-key-1}
\bignorm{(\pmb{D}_h^\pm  J_h^* -J_h^* \pmb{D})(P-z)^{-1}}_{\mathcal{B}(L^2(\re^d),
[\ell^2(h\ze^d)]^{\oplus d})}\leq Ch,
\end{equation}
\begin{equation}\label{eq-elli-key-2}
%\bignorm{\pmb{A}_h J_h^* -J_h^*  \pmb{A}  \remiv{\pmb{D}(P-z)^{-1}}}_{\mathcal{B}(\modiiv{[L^2(\re^d)]^{\oplus d}},
%[\ell^2(h\ze^d)]^{\oplus d})}\leq C h^\a, 
\bignorm{(\pmb{A}_h J_h^* -J_h^*  \pmb{A})  \pmb{D}(P-z)^{-1}}_{\mathcal{B}(L^2(\re^d),
[\ell^2(h\ze^d)]^{\oplus d})}\leq C h^\a, 
\end{equation}
\begin{equation}\label{eq-elli-key-3}
\bignorm{(J_h \pmb{D}_h^\pm- \pmb{D} J_h)(P_h^\pm -\bar{z})^{-1}}_{\mathcal{B}(\ell^2(h\ze^d),[L^2(\re^d)]^{\oplus d})}
\leq Ch, 
\end{equation}
and 
\begin{equation}\label{eq-elli-key-4}
\bignorm{(V_h-z)^{-1} J_h^*-J_h^* (V-z)^{-1}}_{\mathcal{B}(L^2(\re^d),
\ell^2(h\ze^d))}\leq C h^\a,
\end{equation}
for $h\in (0,1]$. Combining these with Lemma~\ref{lem-rel-bdd-elliptic}, we conclude
\begin{equation}\label{eq-elli-diff-est}
\bignorm{J_h^*(P-z)^{-1} - (P_h^\pm -z)^{-1} J_h^*}_{\mathcal{B}(L^2(\re^d),\ell^2(h\ze^d))}
\leq Ch^\a,
\end{equation}
and Theorem~\ref{thm-elliptic-convergence} follows as in Section~\ref{sec-standard-lattices}.

We first consider the relative boundedness. We denote
\begin{align*}
Q:=\pmb{D}^* \pmb{A} \pmb{D}
\end{align*}
and
\begin{align*}
Q_h^\pm := (\pmb{D}_h^\pm)^* \pmb{A}_h \pmb{D}_h^\pm .
\end{align*}

\begin{lem}\label{lem-rel-bdd-elliptic}
$Q$ and $V$ are $P$-bounded. Moreover, $Q_h^\pm$ and $V_h$ are $P_h^\pm$-bounded uniformly in $h\in(0,1]$, i.e. 
\begin{equation}\label{eq-Qh-Phbdd}
\sup_{h\in(0,1]}\norm{Q_h^\pm(P_h^\pm-z)^{-1}} <  \infty ,
\end{equation}
\begin{equation}\label{eq-Vh-Phbdd}
\sup_{h\in(0,1]}\norm{V_h^\pm(P_h^\pm-z)^{-1}} <  \infty .
\end{equation}
\end{lem}

\begin{proof}
We only prove \eqref{eq-Vh-Phbdd}. The other statements follow easily from it. 
The proof is quite similar to that of Lemma \ref{lem-rel-bdd}, or Lemmas~2.4, 2.5 of \cite{NaTad}, 
but we need to be slightly more careful since the free Hamiltonian $H_{0,h}$ is replaced by $Q_h^\pm$. 
We include the proof for the completeness. 

We may assume $V\geq1$ without loss of generality.
Then it follows from $V_h\geq0$ and uniform ellipticity of $\pmb{A}$ that
\[
\jap{u,P_h^\pm u} \geq \jap{u,Q_h^\pm u} = \jap{\pmb{D}_h^\pm u,\pmb{A}_h \pmb{D}_h^\pm u} \geq c_0 \norm{\pmb{D}_h^\pm u}^2.
\]
This implies $Q_h^\pm\geq0$ and
\[
\jap{u,P_h^\pm u} \geq \jap{u,V_h u} = \norm{V_h^{1/2} u}^2 .
\]
Thus we have 
\begin{equation}\label{eq-form-bdd-1}
\sup_{h\in(0,1]} \norm{V_h^{1/2} (P_h^\pm)^{-1/2}}<\infty
\end{equation}
and
\begin{equation}\label{eq-form-bdd-2}
\sup_{h\in(0,1]} \norm{\pmb{D}_h^\pm (P_h^\pm)^{-1/2}}<\infty.
\end{equation}
We set $\tilde{V}:=\y * V$, where $\y\in C_c^\infty(\re^d)$ is fixed so that $\y(x)\geq0$, $\supp[\y]\subset\bra{|x|\leq1}$ and $\int \eta(x) dx=1$ hold.
We note that $\tilde{V}$ satisfies for $x\in\re^d$ and $\a\in\ze_+^d$
\begin{equation}\label{eq-tilde-V-1}
c_1^{-1} V(x)\leq \tilde V(x)\leq c_1 V(x), 
\end{equation}
\begin{equation}\label{eq-tilde-V-2}
\bigabs{\pa_x^\a \tilde V(x) }\leq C_\a \tilde V(x)
\end{equation}
with some $c_1>0$ and $C_\a>0$.
\eqref{eq-tilde-V-1} implies that it suffices to show $\tilde V_h$ is $P_h^\pm$-bounded uniformly in $h$.
We write $W_h(x)=\tilde V_h(x)^{1/2}\geq 1$, and compute 
\[
\tilde V_h (P_h^\pm)^{-1}
=(W_h(P_h^\pm)^{-1/2})(W_h(P_h^\pm)^{-1/2})^* + W_h (P_h^\pm)^{-1}[P_h^\pm,W_h](P_h^\pm)^{-1}.
\]
The first term in the right hand side is uniformly bounded thanks to \eqref{eq-form-bdd-1}.
For the second term, we learn
\[
[W_h, P_h^\pm]= [W_h, (\pmb{D}_h^\pm)^* \pmb{A}_h \pmb{D}_h^\pm]
= [\pmb{D}_h^\pm,W_h]^* \pmb{A}_h \pmb{D}_h^\pm + (\pmb{D}_h^\pm)^* \pmb{A}_h [\pmb{D}_h^\pm,W_h]. 
\]
Since $[\pmb{D}_h^\pm,W_h] W_h^{-1}$ is bounded uniformly in $h$ by \eqref{eq-tilde-V-2}, we can show that, with the help of \eqref{eq-form-bdd-2}, the second term is also bounded uniformly in $h$, which concludes the uniform $P_h$-boundedness of $\tilde V_h$.
\end{proof}

We next prepare a standard elliptic estimate. We denote

\[
H_{0;h}=\sum_{j=1}^d D_{h;j}^+ D_{h;j}^- = \sum_{j=1}^d D_{h;j}^- D_{h;j}^+.
\]

\begin{lem}\label{lem-elliptic-estimate}
There are $c_1, c_2>0$ such that 
\[
\norm{P_h^\pm u}^2 \geq c_1 \norm{H_{0;h} u}^2 -c_2 \norm{u}^2
\]
for $u\in \ell^2(h\ze^d)$, $0<h\leq 1$. 
\end{lem}

\begin{proof}
It follows from Lemma~\ref{lem-rel-bdd-elliptic} that it suffices to show the assertion with $P_h^\pm$ replaced by $Q_h^\pm$, that is, we may assume $V_h\equiv 0$.

We consider the $+$-case. 
We note, by the Schwarz inequality, 
\[
\jap{P_h^+ u, H_{0;h}u} \leq \norm{P_h^+ u}\,\norm{H_{0;h}u}
\leq \frac{1}{4\e}\norm{P_h^+ u}^2 +\e\norm{H_{0;h}u}^2
\]
with any $\e>0$. On the other hand, we have 
\begin{align*}
\jap{P_h^+ u, H_{0;h}u} 
&=\sum_{j=1}^d \jap{D_{h;j}^- P_h^+ u, D_{h;j}^- u} \\
&= \sum_{j=1}^d \jap{D_{h;j}^- u,  P_h^+ D_{h;j}^- u} 
+\sum_{j=1}^d \jap{[D_{h;j}^-,P_h^+]u, D_{h;j}^-u} = \mathrm{I} +\mathrm{II}. 
\end{align*}
We compute each terms: 
\begin{align*}
\mathrm{I} &= \sum_{j,k,\ell=1}^d \jap{(D_{h;j}^- u), D_{h;k}^- a_{k,\ell} D_{h;\ell}^+ (D_{h;j}^-u)}\\
&=\sum_{j,k,\ell=1}^d \jap{(D_{h;k}^+ D_{h;j}^- u),  a_{k,\ell} (D_{h;\ell}^+ D_{h;j}^-u)}\\
&\geq c_0 \sum_{j,k=1}^d\jap{D_{h;k}^+ D_{h;j}^- u,  D_{h;k}^+ D_{h;j}^-u} 
=c_0\sum_{j,k=1}^d\jap{D_{h;k}^+ D_{h;k}^- u,  D_{h;j}^+ D_{h;j}^- u} \\
&= c_0 \jap{H_{0;h}u, H_{0;h}u}
=c_0\norm{H_{0;h}u}^2, 
\end{align*}
where we have used the ellipticity in the third line. We note 
\[
[\pmb{D}_h^-, P_h^+]= [\pmb{D}_h^-, \pmb{D}_h^- \pmb{A}_h \pmb{D}_h^+]
=\pmb{D}_h^- [\pmb{D}_h^-, \pmb{A}_h] \pmb{D}_h^+, 
\]
and $[D_{h;j}^-, a_{k,\ell}]$ are uniformly bounded operators by Assumption~\ref{ass-elli}. 
Thus we have 
\begin{align*}
|\mathrm{II}| &\leq C\sum_{j,k,\ell=1}^d \norm{D_{h;k}^+ D_{h;j}^- u}\, 
\norm{D_{h;\ell}^+ u} \leq C' \norm{H_{0;h}u}\, \norm{H_{0;h}^{1/2}u}\\
&\leq C' \norm{H_{0;h}u}(\e' \norm{H_{0;h}u}+(1/4\e')\norm{u})\\
&=C'\e'\norm{H_{0;h}u}^2+(C'/4\e')\norm{H_{0;h}u}\norm{u}\\
&\leq 2C'\e' \norm{H_{0;h}u}^2 +(C'/16\e'^2)\norm{u}^2. 
\end{align*}
Thus we have 
\[
\frac{1}{4\e}\norm{P_h^+ u}^2 +\e\norm{H_{0;h}u}^2
\geq c_0\norm{H_{0;h} u}^2 -2C'\e' \norm{H_{0;h}u}^2 -(C'/16\e'^2)\norm{u}^2. 
\]
We choose $\e,\e'>0$ so that 
\[
c_0/2 -2C'\e' -\e > 0, 
\]
e.g., $\e'= c_0/8C'$, $\e=c_0/5$, and we have 
\[
\norm{P_h^+ u}^2 \geq (2\e c_0)\norm{H_{0;h}u}^2 - (\e C'/4\e'^2)\norm{u}^2, 
\]
and this completes the proof with $c_1=2\e c_0$ and $c_2=\e C'/4{\e'}^2$. 
\end{proof}

\begin{lem}\label{lem-diff-esti}
Let $z\in\co\setminus\re$. Then there is $C>0$ such that 
\[
\bignorm{(\pmb{D}_h^\pm  J_h^* -J_h^* \pmb{D})(H_0-z)^{-1}}_{\mathcal{B}(L^2(\re^d),
[\ell^2(h\ze^d)]^{\oplus d})}\leq Ch
\]
and 
\[
\bignorm{(J_h \pmb{D}_h^\pm- \pmb{D} J_h)(H_{0;h} -\bar{z})^{-1}}_{\mathcal{B}(\ell^2(h\ze^d),[L^2(\re^d)]^{\oplus d})}\leq Ch
\]
uniformly for $h\in (0,1]$, where $H_0=\pmb{D}\cdot\pmb{D}=-\triangle$. 
\end{lem}

\begin{proof}
The proof is essentially the same as that of Lemma~2.3 of \cite{NaTad}, as well as the discussion in Section~\ref{subsec-lattice-proof}. Instead of considering $H_{0;h}(\x)-H_0(\x)$, we consider
$D_{h;j}^\pm(\x) - D_j(\x)$, where 
\[
D_{h;j}^\pm(\x) =\pm h^{-1}(1-e^{\pm 2\pi hi\x_j}), \quad\text{and}\quad 
D_j(\x)=2\pi i \x_j.
\]
It is easy to observe
\[
|D_{h;j}^\pm(\x)-D_j(\x)|\leq C h |\x|^2, 
\]
and $(H_0(\x)-z)^{-1}=O(|\x|^{-2})$ and $(H_{0;h}^\pm(\x)-z)=O(|\x|^{-2})$ 
on $\supp[\hat\f(h\x)]$. These imply the assertions. 
\end{proof}

\begin{cor}\label{cor-elli-diff}
Estimates \eqref{eq-elli-key-1} and \eqref{eq-elli-key-3} hold. 
\end{cor}

\begin{proof}
Since $(H_0-z)(P-z)^{-1}$ is bounded in $L^2(\re^d)$, \eqref{eq-elli-key-1} follows from the first inequality of Lemma~\ref{lem-diff-esti}. Similarly, \eqref{eq-elli-key-3} follows from the second inequality of 
Lemma~\ref{lem-diff-esti} since $(H_{0;h}-z)(P_h^\pm-z)^{-1}$ is uniformly bounded in $L^2(\re^d)$ by virtue of  Lemmas~\ref{lem-rel-bdd-elliptic}~and~\ref{lem-elliptic-estimate}. 
\end{proof}

\begin{proof}[Proof of Theorem~\ref{thm-elliptic-convergence}]
By Lemma~2.6 of \cite{NaTad}, we immediately conclude \eqref{eq-elli-key-4} and 
\[
\norm{\pmb{A}_h J_h^* -J_h^*  \pmb{A} }_{\mathcal{B}(L^2(\re^d),[\ell^2(h\ze^d)]^{\oplus d})}\leq C h^\a,
\quad h\in (0,1], 
\]
with any $\a\in (0,1)$. Since $\pmb{D}(P-z)^{-1}$ is bounded in $L^2(\re^d)$, \eqref{eq-elli-key-2} also follows immediately. Now we have \eqref{eq-elli-key-1} -- \eqref{eq-elli-key-4}, and they imply 
\eqref{eq-elli-diff-est}. We also have 
\[
\bignorm{(1-J_h J_h^*)(P-z)^{-1}}_{\mathcal{B}(L^2(\re^d))}\leq C h^2
\]
as well as Lemma~\ref{lem-1-JJ}. Combining these, we conclude the assertion as in Section~2.3. 
\end{proof}

%%%%%

\end{document}